\documentclass[preprint,11pt]{elsarticle}

\usepackage{graphicx}
\usepackage{setspace}
\usepackage{amsmath} 
\usepackage{amssymb,nicefrac}
\usepackage{amsfonts}
\usepackage{color}
\usepackage{subfigure}
\usepackage[dvips]{epsfig}
\usepackage[dvips]{graphicx}
\usepackage[section]{placeins}
\usepackage{dsfont}
\usepackage{mathrsfs}
\usepackage{sgame}
\usepackage{float}
\usepackage[T1]{fontenc}
\usepackage{aurical}
\usepackage{tikz,pgfplots,verbatim}
\usepackage[linesnumbered,ruled]{algorithm2e}
\usepackage{multirow}
\usepackage{epstopdf}

\journal{Energy and Buildings}

\newtheorem{proposition}{Proposition}[section]

\newenvironment{proof}{\textbf{Proof.}}{$\square$\\}

\definecolor{Yellow}{rgb}{1, 1, 0}
\definecolor{VeryLightGray}{gray}{.90}
\definecolor{LightGray}{gray}{.7}
\definecolor{Gray}{gray}{.50}
\definecolor{DarkGray}{gray}{.3}
\definecolor{VeryDarkGray}{gray}{.10}

\newcommand{\tr}{^{\mathrm T}}

\newcommand{\magn}[1]{\left\vert #1 \right\vert}

\DeclareGraphicsExtensions{.pdf,.jpeg,.png,.jpg}  

\begin{document}

\begin{frontmatter}



\title{Generalized Online Transfer Learning for Climate Control in Residential Buildings\tnoteref{t1,t2}}

\tnotetext[t1]{The research reported in this paper has been supported by the Austrian Ministry for Transport, Innovation and Technology, the Federal Ministry of Science, Research and Economy, and the Province of Upper Austria in the frame of the COMET center SCCH.}
 \tnotetext[t2]{An earlier version of part of this paper appeared in \cite{grubinger2016}.}
    


\author[scch]{Thomas Grubinger}
\ead{thomas.grubinger@scch.at}
   
\author[scch]{Georgios C. Chasparis\corref{cor1}}
\ead{georgios.chasparis@scch.at}
\author[scch]{Thomas Natschl\"{a}ger}
\ead{thomas.natschlaeger@scch.at}

\address[scch]{Department of Data Analysis Systems, Software Competence Center Hagenberg GmbH, Softwarepark 21, A-4232 Hagenberg, Austria}

\cortext[cor1]{Corresponding author}

\begin{abstract}
This paper presents an online transfer learning framework for improving temperature predictions in residential buildings. In transfer learning, prediction models trained under a set of available data from a target domain (e.g., house with limited data) can be improved through the use of data generated from similar source domains (e.g., houses with rich data). Given also the need for prediction models that can be trained online (e.g., as part of a model-predictive-control implementation), this paper introduces the generalized online transfer learning algorithm (GOTL). It employs a weighted combination of the available predictors (i.e., the target and source predictors) and guarantees convergence to the best weighted predictor. Furthermore, the use of Transfer Component Analysis (TCA) allows for using more than a single source domains, since it may facilitate the fit of a single model on more than one source domains (houses). This allows GOTL to transfer knowledge from more than one source domains. We further validate our results through experiments in climate control for residential buildings and show that GOTL may lead to non-negligible energy savings for given comfort levels.
\end{abstract}


\end{frontmatter}



\section{Introduction}	\label{sec:Introduction}

Recent studies on climate control (heating/cooling) in residential buildings have demonstrated the importance of accurate system identification and prediction for energy savings \cite{NghiemPappas11,Oldewurtel12,TouretzkyBaldea13,ChasparisNatschlaeger16,Martinez14}. In fact, there have been several efforts on exploiting the benefits of such prediction schemes through the development of \emph{model-predictive-control} (MPC) approaches \cite{NghiemPappas11,Oldewurtel12,TouretzkyBaldea13}, where predictions of the temperature evolution, weather conditions and user behavior can be incorporated directly into the control design. 

The current trend on system identification and prediction in a residential building (\emph{target house}) exploits measurements collected during normal operation of the heating/cooling system. Several identification schemes have been used to generate predictions, including the MIMO ARMAX model \cite{YiuWang07}, ARX models \cite{Malisani10} and the neural network approach \cite{Mustafaraj11}. 

Although linear transfer function models are the most commonly used models for system identification in residential buildings (due to the resulting simplified MPC design), changes in weather conditions and/or the heating patterns may give rise to nonlinear phenomena and consequently to variations in the prediction performance. This observation has been pointed out by several authors, leading to more detailed identification schemes, such as 
the multiple-time scale analysis presented in \cite{TouretzkyBaldea13}, the more detailed models of HVAC systems discussed in \cite{Scotton13} and the nonlinear regression models developed in \cite{ChasparisNatschlaeger16}.

Variations in the weather conditions and/or heating patterns may always occur throughout the year. The reliability of the (\emph{target house's}) prediction model may be improved by incorporating the formulation of data/prediction models from other residential buildings (\emph{source houses}). To this end, this paper addresses the following question: \emph{In what form such a ``knowledge transfer'' from a source house to a target house should be performed and under which conditions it could be beneficial in terms of the resulting prediction performance?} 

Such knowledge transfer objective is usually encountered in machine learning and it may take alternative forms depending on the application itself. In particular, knowledge transfer can usually be performed within the context of \emph{Transfer Learning}~\cite{pan2010survey}. Generally, transfer learning aims at transferring knowledge from previous (source) tasks to a target task when the latter has limited training data. Transfer learning has received a lot of attention in recent years and has successfully been used in several applications, such as indoor localization \cite{pan_transfer_2008}, image processing \cite{hinton2007using}, land-mine detection \cite{grubinger2015domain} and biological applications \cite{muandet2013domain}.

Most transfer learning approaches can be referred to as \emph{offline approaches}, since learning is performed offline. However, in the context of climate control in residential buildings, data are usually collected continuously and knowledge transfer needs to be implemented in an \emph{online} fashion. Thus, it is necessary to develop an online transfer learning methodology that will be particularly appropriate for knowledge transfer between different houses. To the best of our knowledge, only the \emph{Online Transfer Learning (OTL)}~\cite{zhao2010otl} method addresses an online learning case. It uses a weighted prediction of an offline classifier (learned on the source scenario data) and an incrementally updated online classifier on the target scenario data. Weighted predictors are also common in \emph{ensemble learning} methods \cite{dietterich2000ensemble,fern2003online,fan1999application,littlestone1988learning,littlestone1994weighted}, however predictors are constructed from a \emph{single} dataset (thus, they are not directly related to transfer learning). An online algorithm for the case of \emph{multitask learning}~\cite{evgeniou2007multi} was introduced by Dekel~{\em et al.}~\cite{dekel2007online}. In contrast to transfer learning, multitask learning addresses the problem of learning different tasks in parallel. Furthermore, Dekel's method is designed for classification tasks, while climate control requires regression tasks. 

Given the absence of online transfer learning methodologies that can be directly applied for knowledge transfer for climate control in residential buildings, \emph{this paper develops an online transfer learning algorithm that is particularly appropriate for climate control}. In particular, we introduce \emph{generalized online transfer learning} (GOTL) that is based on a weighted combination of: (1)~an offline model (linear regression model learned on the source data (with data collected over an extended time horizon), and (2)~an online regressor (recursive least squares, cf.,~\cite{Sayed03}) learned on the target house -- which is incrementally updated as new data arrives. The proposed algorithm is related to the OTL algorithm of \cite{zhao2010otl}. However, in \cite{zhao2010otl}, the weighted predictor is only eligible for classification, while our framework is applicable for both classification and regression. Furthermore, our online transfer learning scheme guarantees convergence to the globally optimal weights for the predictors. 

The proposed GOTL algorithm is not limited to the use of a single source house. Clearly, standard machine learning methods cannot be applied directly on the combined data of different source houses, as the different datasets were not necessarily generated by the same process. However, (offline) {\em domain generalization}~\cite{muandet2013domain} methods can be used to jointly model data from different source domains. In this paper, we use the domain generalization method {\em Transfer Component Analysis (TCA)}~\cite{pan2011domain,grubinger2015domain}. TCA and domain generalization methods in general, aim to transform data from different domains into a shared subspace, where the distributions of the domains are similar. Once transformed, a (linear) regression model can be constructed on the joint source house data.


Our contributions can be summarized as follows: (1) We propose an online transfer learning methodology (GOTL) that is appropriate for knowledge transfer between residential buildings. To the best of our knowledge, this is the first attempt to address knowledge transfer in this application domain. 
(2) The proposed algorithm builds upon the online transfer learning methodology of \cite{zhao2010otl} and (a) guarantees convergence to the global optimum combination, (b) it addresses both classification and regression tasks. (3) We demonstrate through experiments that is at least as good as either one of the source house and target house predictor. (4) We show that the domain generalization technique TCA can be used to learn a joint model from several source houses, which proves to be even more beneficial for the target task. (5) We demonstrate that the improvement in predictive accuracy translates into non-negligible energy savings for given comfort levels. This paper extends  previous work of the authors~\cite{grubinger2016} by contributions (4) and (5) and provides a more detailed presentation of the already published work.


In the remainder of the paper, Section~\ref{sec:Framework} provides the framework and the main objective of this paper. Section~\ref{sec:OnlineTransferLearning} describes the proposed online transfer learning methodology (GOTL) from a single source domain. In Section~\ref{sec:OTLfromMultipleSourceDomains}, we extend our methodology to the case of multiple source domains through the use of TCA. Section~\ref{sec:ExperimentalSetup} provides a description on the experimental framework and Section~\ref{sec:Results} demonstrates the experimental results. Finally, Section~\ref{sec:Conclusions} presents concluding remarks and future work.


\section{Framework \& Objective}	\label{sec:Framework}

In this paper, we are interested in the development of a \emph{prediction (input-output transfer) model of the heat-mass transfer dynamics in residential buildings}, that also exploits data collected from other (not necessarily similar) houses. Such prediction models can be used within an MPC and provide predictions of the indoor temperature over an optimization horizon of interest (e.g., several hours ahead). 

A prediction (input-output transfer) model of the heat-mass transfer dynamics in buildings can be formulated in the following generic form:
\begin{eqnarray}	\label{eq:GenericPredictor}
\hat{y}_t & = & f(y_{t-1},...,y_{t-\ell},\mathbf{u}_{t-1},...,\mathbf{u}_{t-\ell}),
\end{eqnarray}
where $y\in\mathbb{R}$ denotes the variable of interest, $\hat{y}\in\mathbb{R}$ denotes the estimate of $y$ and $\mathbf{u}\in\mathbb{R}^{1\times{m}}$ denotes the control inputs/disturbances. The predictor $f:\mathbb{R}^{\ell + \ell m}\to\mathbb{R}$ is a linear or nonlinear prediction model. 

Under the assumption that measurements are perturbed by a white measurement noise, predictors of the form (\ref{eq:GenericPredictor}) correspond to the \emph{maximum a posteriori} predictor (as in the case of the Output-Error regression model \cite[Chapter~4]{Ljung99}). For example, in the context of climate control in buildings, $y$ corresponds to the indoor temperature that needs to be regulated, while $\mathbf{u}$ may include all variables directly or indirectly affecting the indoor temperature, such as the flow of the thermal medium, the inlet/outlet temperatures of the thermal medium, the occupants presence, the outdoor temperature and the solar gain.

Let us assume that measurements of the inputs and output are collected at regular time instances $T_s,2T_s,...$, briefly denoted by $t=1,2,...$, where $T_s$ corresponds to the sampling period. Assuming that an MPC is implemented for temperature control, let $T_{\rm hor} \doteq M T_s$ denote our optimization horizon over which predictions are requested for some large $M \in \mathbb{N}$. We denote the time instances at which predictions are requested by $t_k$, $k=1,2,...$ such that $t_k = k M$. Note that $t_k$ is a subsequence of the time index $t$, thus predictions are requested over $M,2M,...$ time instances. We will often refer to these time instances as the \emph{evaluation/optimization instances}. 

To minimize notation, we briefly denote
\begin{equation} \label{eq:GenericX}
\mathbf{x}_t \doteq \left(\begin{array}{cccccc}
y_{t-1} & \cdots & y_{t-\ell} & \mathbf{u}_{t-1} & \cdots & \mathbf{u}_{t-\ell} 
\end{array}\right).
\end{equation}
Since we will be concerned with providing predictions at time instances $t_k$, $k=1,2,...$, we will compactly denote the data available at those time instances by 
\begin{equation}
\mathbf{X}_{k} \doteq \left[\begin{array}{c}
\mathbf{x}_0 \\
\mathbf{x}_1 \\
\vdots \\
\mathbf{x}_{t_k}
\end{array}\right],
\end{equation}
for all $k=1,2,...$ which correspond to time instances $t_k$.

In the remainder of the paper, we will assume that data are available both from a \emph{target} house, where we wish to minimize the prediction error, and one or more \emph{source} house(s), denoted by the symbol $S$. We assume that the input-output variables measured in both the source house(s) and target house are the same. Let $\mathbf{X}_{S}$ and $\mathbf{X}_{k}$ be the corresponding \emph{feature data} available at $t_k$. Since we are concerned with \emph{transfer learning}, the data set $\mathbf{X}_{S}$ is assumed a-priori fixed. However, the forthcoming analysis can be extended in a straightforward manner to the case that this data set also varies with time.

We introduce two prediction models: (1) $f_{k}(\cdot)$ is a prediction model that is trained online at time instances $t_k$, $k=1,2,...$ over the currently available data $\mathbf{X}_{k}$ from the target house. (2) $f_{S}(\cdot)$ denotes a prediction model that is trained offline over the a-priori available data set $\mathbf{X}_{S}$. In case of a single source house the predictor is taking the form $f_{S}(\cdot) = h_{S}(\cdot)$, where $h_{S}(\cdot)$ is a supervised prediction function. In case of multiple source houses $f_{S}(\cdot) = h_{S}(\theta(\cdot))$, where $\theta(\cdot)$ is a domain generalization function, which maps the data from the different houses into a shared subspace. Note that in this case, the source data are $\mathbf{X}_{S} = \mathbf{X}_{S_1}\cup \mathbf{X}_{S_2} \cup ... \cup \mathbf{X}_{S_D}$, where $D$ is the number of source houses. The definition of the domain generalization function and its role will be discussed in detail in the forthcoming Section~\ref{sec:GOTLMultipleSourceDomains}.


The goal of this paper is the computation of a new (combined) predictor for the \emph{target} house, evaluated at time instances $t_k$, $k=1,2,...$, which admits the generic form $F_{k}(\cdot)=F_{k}(f_{k}(\cdot),f_{S}(\cdot)).$ 
We wish to address the following optimization problem:
\begin{equation}	\label{eq:ProblemFormulation}
\min_{\{F_{k}\}} \quad \sum_{k=1}^{K}\sum_{t=t_{k}+1}^{t_{k}+M}\delta^{t_k+M-t}\magn{F_{k}(\mathbf{\hat{x}}_{t|k})-{y}_{t}}^2,
\end{equation}
on the target house, where $K$ is the number of evaluation intervals considered and $\delta \in (0,1]$ is a discount factor. As in an MPC implementation, the measurements $y$ are only available at the beginning of each evaluation interval $k$, i.e., at time $t_k$, while the performance of a prediction model $F_k$ has to be evaluated over an optimization horizon of $M$ steps ahead. Thus, an estimate $\hat{\mathbf{x}}_{t|k}$ of $\mathbf{x}_{t}$ is used in the formulation of the predictions, which is defined as 
\begin{equation*}
\mathbf{\hat{x}}_{t|k} \doteq \left(\begin{array}{cccccc}
\tilde{y}_{t-1|k} & \cdots & \tilde{y}_{t-\ell|k} & \mathbf{\tilde{u}}_{t-1|k} & \cdots & \mathbf{\tilde{u}}_{t-\ell|k} 
\end{array}\right),
\end{equation*} 
\begin{equation*}
\mbox{where} \qquad
\tilde{y}_{t|k} \doteq \begin{cases} 
y_t  & \mbox{if } t \le t_k \\
\hat{y}_{t|k} & \mbox{if } t  >  t_k.
\end{cases}
\end{equation*} 
Here, $\hat{y}_{t|k}$ are generated as $\hat{y}_{t|k}=F_k(\hat{\mathbf{x}}_{t|k}),$ where the prediction at time instance $t$ are given by the prediction function $f_k(\cdot)$, which has been trained using all data up to time instance $t_k$. Regarding the estimates of the control inputs and the exogenous disturbances summarized in $\tilde{\mathbf{u}}_{t|k}$, we consider perfect estimates, since we would like to investigate the prediction performance of $F_k$ over $y$. Therefore, we set $\tilde{\mathbf{u}}_{t-j|k} = {\mathbf{u}_{t-j}}, j \in \{1,...,l\}$.

\section{Online Transfer Learning from Single Source Domain}		\label{sec:OnlineTransferLearning}

\subsection{Generalized Online Transfer Learning (GOTL)} 	\label{sec:GOTL}

Obviously, when addressing an optimization problem of the form (\ref{eq:ProblemFormulation}), the optimal choice of a combined predictor may not be computed a-priori, i.e., before receiving the measurements $\{y_t\}_{t=t_1}^{t_k}$ from the target house. Thus, an online optimization scheme is required. Besides, the choice of an optimal predictor may change frequently with time, hence requiring frequent revisions of the combined predictor. 

In this paper, we propose a \emph{transfer learning} algorithm that addresses the generic problem formulation of (\ref{eq:ProblemFormulation}) in an \emph{online} fashion under the structural constraint of the form:
\begin{equation}	\label{eq:WeightedPredictor}
F_{k}(\mathbf{x}_t;\alpha_k) \doteq (1-\alpha_k) f_{k}(\mathbf{x}_t) + \alpha_k f_{S}(\mathbf{x}_t),
\end{equation}
$\mbox{where } \alpha_k\in\mathcal{A}\doteq\left\{0,\Delta,2\Delta,...,1-\Delta,1\right\},$ is a weight assigned to the source predictor. The constant $\Delta\in(0,1)$ is selected so that $\Delta = \nicefrac{1}{n}$ for some large $n\in\mathbb{N}$. 

In other words, we consider combined predictors that can be represented as a weighted sum of the two available predictors (the source predictor trained offline and the target predictor trained online). In this case, the optimization problem (\ref{eq:ProblemFormulation}) can be translated to an optimization problem over $\{\alpha_k\}$ and takes on the following form:
\begin{equation}	\label{eq:ProblemFormulationWeightedSum}
\min_{\{\alpha_k\}} \sum_{k=1}^{K}\sum_{t=t_{k}+1}^{t_{k}+M}\delta^{t_{k}+M-t}\magn{F_{k}(\hat{\mathbf{x}}_{t|k};\alpha_k)-{y}_{t}}^2.
\end{equation}

The proposed algorithm \emph{Generalized Online Transfer Learning} (GOTL) is motivated by the so-called \emph{adaptive learning} \cite{Young93} defined in games and it is based on the notion of \emph{better reply}. It is described in detail in Table~\ref{Tb:GOTL}. 

\begin{table}[t!]
\fbox{
\begin{minipage}{0.9\textwidth}
At the end of any evaluation interval $k$ (i.e., at time instance $t_{k+1}$), $k=1,2,...,K$:
\begin{enumerate}
  \item Update the history of feature data $\mathbf{X}_{k+1}$. 
  \item Evaluate the currently selected weight $\alpha_{k}$ by computing its \emph{better reply}
  \begin{equation*}
  {\rm B}_{k}(\alpha_k) \doteq \{ \alpha'\in\mathcal{A}(\alpha_k): R_{k}(\alpha') < R_{k}(\alpha_k) \}.
  \end{equation*}

  \item Select a new weight $\alpha_{k+1}$ according to 
  \begin{equation*}
  \alpha_{k+1} \in \begin{cases}
  {\rm B}_k(\alpha_k) & \mbox{if } {\rm B}_k(\alpha_k) \neq \varnothing \\
  \alpha & \mbox{else.}
  \end{cases}
  \end{equation*}
  
  \item Train the online predictor for the target $f_{k}(\cdot)$ by using the updated history $\mathbf{X}_{k+1}$, and get $f_{k+1}(\cdot)$.
  
  \item Define the new combined predictor $F_{k+1}$ for creating predictions over $t\in\{t_{k},...,t_{k+1}\}$ as follows: 
  \begin{eqnarray*}  
   F_{k+1}(\cdot;\alpha_{k+1}) \doteq (1-\alpha_{k+1})f_{k+1}(\cdot) + \alpha_{k+1}f_{S}(\cdot).
  \end{eqnarray*}

  \item Update the time $k\leftarrow{k+1}$ and repeat.

\end{enumerate}
\end{minipage}
} 
\caption{Generalized Online Transfer Learning (GOTL)} \label{Tb:GOTL}
\end{table}

In particular, at the end of every evaluation interval $k=1,2,...$, the current combined predictor $F_{k}$, defined in (\ref{eq:WeightedPredictor}) and employing weight $\alpha_{k}$, is evaluated over the updated history of measurements $\mathbf{X}_{k+1}$, i.e., the measurements collected at the end of the evaluation interval $k$.

Its performance with respect to the prediction error is then compared with the corresponding performances when the weight $\alpha_k$ is slightly perturbed. In particular, the performance of the weight $\alpha_k$ is compared with the corresponding performances of the weights selected from the set $\mathcal{A}(\alpha_k)$ defined as follows:
\begin{equation*}
\mathcal{A}(\alpha_k) \doteq \begin{cases}
\{\alpha_k-\Delta,\alpha_k,\alpha_k+\Delta\} & \mbox{if $\Delta< \alpha_k < 1 - \Delta$}\\
\{\alpha_k,\alpha_k+\Delta\} & \mbox{if $\alpha_k \leq \Delta$} \\
\{\alpha_k-\Delta,\alpha_k\} & \mbox{if $\alpha_k \geq 1-\Delta$}. 
\end{cases}
\end{equation*}
Comparison between the alternative weights in the set $\mathcal{A}$ is performed with respect to the \emph{discounted weighted average squared error}, $R_{k}$, of the combined predictor, defined as
\begin{equation} 	\label{eq:MSE}
R_{k}(\alpha) \doteq \sum_{j=1}^{k} \sum_{t=t_j+1}^{t_j+M}{\delta^{t_j+M-t}\left|F_{j}(\mathbf{\hat{x}}_{t|j};\alpha)-y_t\right|^2},
\end{equation} 
for some $\delta\in(0,1]$. Note that the predictions $F_{j}(\mathbf{\hat{x}}_{t|j};\alpha)$ (in the evaluation interval $j$) are generated using the target predictor $f_j(\cdot)$ which has been trained using data up to time $t_j$, $\mathbf{X}_j$. The predictions of $F_{j}(\mathbf{\hat{x}}_{t|j};\alpha)$ are evaluated over the time interval $t_1,...,t_{j+1}$. Lastly, note that the trained predictor $f_j(\cdot)$ does not depend on $\alpha$, which implies that it does not have to be refitted when $\alpha$ is updated.

The proposed scheme is related to the \emph{online transfer learning} (OTL) algorithm \cite{zhao2010otl}. In reference \cite{zhao2010otl} the weight update is only applicable for classification problems, while GOTL's weight update mechanism is more general to accommodate both regression and classification problems.

\subsection{Convergence behavior}

Let us define the set of \emph{locally-optimal weights} $$\mathcal{A}_k^* \doteq \left\{\alpha\in\{0,\Delta,...,1-\Delta,1\}: {\rm B}_k(\alpha) = \varnothing\right\}.$$ We can show that GOTL converges to a weight in $\mathcal{A}_{k}^*$, as long as the set $\mathcal{A}_{k}^*$ changes sufficiently slow with $k$ as the following proposition shows.
\begin{proposition}[Convergence to Local Minima]		\label{Pr:ConvergenceToLocalMinima}
\textit{For any weight update instance $k^*$, if $\mathcal{A}_{k}^* = \mathcal{A}^* \neq \varnothing$ for every $k=k^*,k^*+1,...,k^*+n$, then $\alpha_{k^*+n} \in\mathcal{A}^*.$}
\end{proposition}
In other words, if the set of locally optimal weights does not change within the next $n$ update steps, then the process will reach a weight within this set.

\begin{proof}
The proof is a direct implication of the definition of the ${\rm B}_k(\cdot)$, since at most $n$ update steps are required for the process to approach a weight $\alpha^*\in\mathcal{A}^*$ starting from any initial weight.
\end{proof}

\begin{proposition}[Unique Minimizer]		\label{Pr:UniqueMinimizer}
\textit{For any update instance $k$,\footnote{For any finite set $\mathcal{A}$, $\magn{\mathcal{A}}$ denotes its cardinality.} $\magn{\mathcal{A}_k^*}=1.$ Furthermore, this is the unique globally optimal weight.}
\end{proposition}
\begin{proof}
Let us denote the estimation error functions: $e_{j}(t) \doteq f_{j}(\hat{\mathbf{x}}_{t|j}) - y_t,$ and $e_{S,j}(t) \doteq f_{S}(\hat{\mathbf{x}}_{t|j}) - y_t$ for the target and the source predictor, respectively, when evaluated over the evaluation interval $j=1,2,...$. Let also $\mathbb{E}_j\{\cdot\}$ denote the discounted weighted average of the values of a random variable evaluated over the evaluation interval $j$, i.e., $\mathbb{E}_j\{e_j\}\doteq\sum_{t=t_j+1}^{t_j+M}\delta^{t_j+M-t}\{e_j(t)\}.$ Then, the optimization problem $\min_{\alpha\in\mathcal{A}}R_{k}(\alpha)$ can equivalently be written as:
\begin{equation}	\label{eq:PerformanceEquivalentCriterion}
\min_{\alpha\in\mathcal{A}} \sum_{j=1}^{k}\mathbb{E}_j\left\{\magn{(1-\alpha) e_{j} + \alpha e_{S,j}}^2 -  
\magn{e_{j}}^2\right\},
\end{equation}
since the estimation error of the target predictor, $e_{j}(t)$, is independent of $\alpha$. Note further that:
\begin{eqnarray*} 
\lefteqn{\magn{(1-\alpha) e_{j} + \alpha e_{S,j}}^2 - \magn{e_{j}}^2 = } \cr &&
\alpha^2 e_{S,j}^2 - \alpha(2-\alpha)e_{j}^2 + 2\alpha(1-\alpha)e_{S,j}e_{j}.
\end{eqnarray*}
Thus, the optimization (\ref{eq:PerformanceEquivalentCriterion}) can equivalently be written as:
\begin{eqnarray*}	\label{eq:PerformanceEquivalentCriterion2}
\min_{\alpha\in\mathcal{A}} \Big\{ \alpha^2\sum_{j=1}^{k}\mathbb{E}_j\left\{e_{S,j}^2\right\} - \alpha(2-\alpha) \sum_{j=1}^{k}\mathbb{E}_j\left\{e_{j}^2\right\} + \cr 2\alpha(1-\alpha) \sum_{j=1}^{k}\mathbb{E}_j\left\{e_{j}e_{S,j}\right\} \Big\}.
\end{eqnarray*}
Note that the latter objective function is quadratic with respect to $\alpha$ and its second gradient with respect to $\alpha$ is nonnegative. 
This implies that the optimization $\min_{\alpha\in\mathcal{A}}R_k(\alpha)$ admits a unique minimizer.
\end{proof}

It is also straightforward to check that under the hypotheses of Proposition~\ref{Pr:ConvergenceToLocalMinima}, the weight update approaches the unique minimizer after a finite number of steps.

\subsection{Discussion}

Note that the unique minimizer may change with $k$, depending on the conditions under which the data have been collected. However, according to Proposition~\ref{Pr:ConvergenceToLocalMinima}, it is straightforward to check that as long as $\mathcal{A}_k^*$ changes sufficiently slowly with respect to $k$, \emph{the weight update will always approach the (current) minimizer of (\ref{eq:MSE})}.

The question that naturally emerges is under which conditions such a combined predictor will provide a better estimate compared to the target predictor. Let us use the definition of $\mathbb{E}_j\{\cdot\}$ in the proof of Proposition~\ref{Pr:UniqueMinimizer}. Note that the optimization problem $\min_{\alpha\in\mathcal{A}}R_{k}(\alpha)$ is equivalent to (\ref{eq:PerformanceEquivalentCriterion}), which includes the following expectation:
\begin{eqnarray*}
\lefteqn{\mathbb{E}_j\left\{\magn{(1-\alpha) e_{j} + \alpha e_{S,j}}^2 - \magn{e_{j}}^2\right\} = } \cr
&& \alpha^2\mathbb{E}_j\{e_{S,j}^2\} - \alpha(2-\alpha)\mathbb{E}_j\{e_{j}^2\} + 2\alpha(1-\alpha)\mathbb{E}_j\{e_{S,j}e_{j}\} 
\end{eqnarray*}
For the combined predictor to provide smaller prediction error compared to the target predictor, the above quantity has to be strictly negative. If the ``product bias'' is of the same sign, i.e., $\mathbb{E}_j\{e_{S,j}e_{j}\}\geq{0}$, then $\alpha$ should be sufficiently close to one for the above quantity to be negative. On the other hand, if $\mathbb{E}_j\{e_{S,j}e_{j}\}<{0}$, then an $0<\alpha<1$ may improve the prediction error 
even when $\mathbb{E}_j\{e_{S,j}^2\}$ and $\mathbb{E}_j\{e_{j}^2\}$ are of similar size, i.e., even when the source and target predictor perform equally well on the target data. 

%

\section{Online Transfer Learning from Multiple Source Domains}    \label{sec:OTLfromMultipleSourceDomains}

\subsection{Transfer Component Analysis (background)}		\label{sec:DomainAdaptation}

{\em Domain generalization}~\cite{blanchard2011generalizing,muandet2013domain,grubinger2015domain} addresses mismatches between different input distributions. Across-domain information is extracted from the source domain data (where training data is available) and can be used on the target domains (where no training data is available) without re-training. {\em Transfer Component Analysis (TCA)}~\cite{pan2011domain,grubinger2015domain} is a popular domain generalization technique that aims to learn a shared subspace between different domains. In the shared subspace, the data distributions of different domains should be close to each other and task-relevant information of the original data should be preserved. In addition to the distribution matching, TCA preserves the properties of the data.\footnote{This is achieved by maximally preserving the data variance, similarly to \textit{Principal Component Analysis} (PCA) and \textit{Kernel-PCA}~\cite{Scholkopf:1999KernelPCA}.} 
Furthermore, any machine learning method for regression, classification or clustering can be used on the identified subspace. Although, in its original version by Pan~\textit{et al.}, TCA was introduced for two domains \cite{pan2011domain}, here we utilize an extension to multiple source domains introduced by Grubinger \textit{et al.}~\cite{grubinger2015domain}.

Consider the setting where measurements are available from multiple source houses $i \in 1,\dots,D$, denoted by $\mathbf{X}_{S_i}$ and a target house $\mathbf{X}_{T}$. TCA is applicable if $P(\mathbf{X}_{S_i}) \neq P(\mathbf{X}_{S_j}), 1 \leq i < j \leq D$, where $P(\mathbf{X}_{S_i})$ is the probability distribution of $\mathbf{X}_{S_i}$. \emph{The goal of TCA is to find a kernel-induce feature map $\Psi$ such that $P(\Psi(\mathbf{X}_{S_i})) \approx P(\Psi(\mathbf{X}_{S_j}))$} and $P(\Psi(\mathbf{X}_{S_i})) \approx P(\Psi(\mathbf{X}_{T}))$. Once transformed, the combined source and target data can be used in the subsequent machine learning task. Note though that no data from the target house $\mathbf{X}_{T}$ is available for the construction of $\Psi(\cdot)$. The assumption of domain generalization is that the source and target domains are related such that common information can be learned from the source domains and applied to the target domain.

The goal of matching the probability distributions of the available source domains can be translated into a mathematical program as the following derivation demonstrates. 

Let $\mathbf{K}$ be a combined Gram matrix~\cite{Scholkopf:1999KernelPCA} of the cross-domain data of the source domain datasets $\mathbf{X}_{S_1}\cup \mathbf{X}_{S_2} \cup ... \cup \mathbf{X}_{S_D}$. Each element $K_{i,j}$ of $\mathbf{K} \in \mathbb{R}^{N \times N}$ is given by $\phi(\mathbf{x}_i)\tr \phi(\mathbf{x}_j)$, where $N$ is the total number of instances of the source domain datasets and $\phi$ is a kernel function. Note that, with the employment of a linear kernel function in this works the construction of $K_{i,j}$ simplifies to $\phi(\mathbf{x}_i)\tr \phi(\mathbf{x}_j) = \mathbf{x}_i\tr \mathbf{x}_j$. Let the elements $L_{i,j}$ of $\mathbf{L} \in \mathbb{R}^{N \times N}$ be defined as
\begin{equation}
  {L}_{{i,j}} = \left\{
  \begin{array}{l l}
     \frac{S-1}{N^2n_s^2} & \quad \text{if $\mathbf{x}_i,\mathbf{x}_j \in \mathbf{X}_{S_d}$}\\
    -\frac{1}{N^2n_s n_u} & \quad \text{if $\mathbf{x}_i \in \mathbf{X}_{S_d}, \mathbf{x}_j \in \mathbf{X}_{S_u}$ and $d \neq u$}
  \end{array} \right.
	\label{eq:L}
\end{equation}
where $d,u \in \{1,...,D\}$. With parameter matrix $W \in R^{N \times m}, m \ll N$ and the tradeoff parameter $\mu \geq 0$ for the complexity term $tr(\mathbf{W}^T\mathbf{W})$, the objective of TCA is defined as~\cite{pan2011domain} 
\begin{equation}
	\min_{\textbf{W}} tr(\mathbf{W}\mathbf{W}\tr\mathbf{KLKW}) + \mu \ tr(\mathbf{W}\tr\mathbf{W}), \ \text{s.t.} \  \mathbf{W}\tr\mathbf{KHKW} = \mathbf{I}.
\label{eq:opt}
\end{equation}
Here, the centering matrix $\mathbf{H}$ is defined as $\mathbf{H} = \mathbf{I} - \frac{1}{N}\mathbf{11}\tr$, where $\mathbf{1}\in \mathbb{R}^{N}$ is a column vector with all ones and $\mathbf{I} \in \mathbb{R}^{N \times N}$ is the identity matrix. $tr(\mathbf{W}\tr\mathbf{KLKW})$ corresponds to the mismatch of distribution by the \textit{maximum mean discrepancy (MMD)}~\cite{gretton2006kernel} distance and $\mathbf{W}\tr\mathbf{KHKW}$ is the variance of the projected samples. The embedding of the data in the latent space is given by $\mathbf{W}\tr\mathbf{K}$. As shown by Pan~\textit{et al.}~\cite{pan2011domain}, the solution of $\mathbf{W}$ is given by the $m \ll N$ leading eigenvectors of  
\begin{equation}
(\mathbf{KLK}+\mu \mathbf{I})^{-1}\mathbf{KHK}.
\label{eq:sol}
\end{equation}
Parameter $m$ is commonly selected by cross validation~\cite{pan2011domain,grubinger2015domain}.

\subsection{GOTL under multiple source domains}     \label{sec:GOTLMultipleSourceDomains}

In general, GOTL can be applied to multiple source domains following the same outline as described in Section~\ref{sec:GOTL}. The only difference is the use of $f_{S}(\cdot)$. For the case of a single source house $f_{S}(\cdot) = h_{S}(\cdot)$, where $h_{S}(\cdot)$ is a supervised prediction function. In case of multiple source domains $f_{S}(\cdot) = h_{S}(\theta(\cdot))$, where $\theta(\cdot)=\mathbf{W}^T[\phi(\cdot)^T\phi(\mathbf{x}_j)]_{j \in N}$.

\section{Experimental Setup}	\label{sec:ExperimentalSetup}

In this section, we describe the experimental setup with which the proposed GOTL algorithm was tested for climate control in residential buildings. 

\subsection{Simulation platform}

We used a standard tool for modeling and simulating residential buildings, namely EnergyPlus (V7-2-0) developed by the U.S. Department of Energy \cite{EnergyPlus}. The Building Controls Virtual Test Bed (BCVTB) simulation tool has also been used for allowing data collection and also climate control developed in MATLAB to be implemented during run-time. A three-storey residential building was modeled and simulated with the EnergyPlus environment to allow for collecting data from a realistic residential environment.

\subsection{Data Generation}		\label{sec:DataGeneration}


The simulated house is equipped with a radiant heating system which operates under an intermittent operation (i.e., on/off) pattern. \emph{We are concerned with the prediction and control of the temperature of a single thermal zone of this house}. When evaluating the prediction accuracy of the developed prediction models, the data were generated under a standard hysteresis controller with set temperature equal to $21^oC$ and sampling period $T_{s}=\nicefrac{1}{2}h$ (cf.,~\cite[Section~5.2]{ChasparisNatschlaeger16}), thus assuming normal operating conditions. When, instead, the developed (online) prediction model was also used to provide predictions under an MPC formulation, the training data were generated under the currently implemented MPC. The details of these experiments will be discussed in a forthcoming section  (Section~\ref{sec:PredictionPerformanceComparison}). 

Independently of the data generation process (offline or online), there also exists a natural ventilation system that operates autonomously and with a constant air flow, i.e., there is no heating control through the HVAC system. The following parameters can be measured: the temperature of all thermal zones; the outdoor temperature; the water flow and the inlet water temperature of the radiant heating system; and all exogenous heat sources, namely the solar gain and the occupants presence. 

Lastly, it is important to note that for the low-order linear models identified in this paper, the experiments (offline or online) will be \emph{informative} (cf.,~\cite[Section~5.5]{ChasparisNatschlaeger16}). This is due to the intermittent pattern of both the input and disturbance signals, which result in a sufficient number of distinct frequencies in these signals. In particular, even though the experiments presented in this paper are closed-loop experiments (since the water flow of the radiant heating system depends on the zone temperature), the input signal of the water flow is a nonlinear function of both the occupancy pattern as well as the zone temperature. Thus, as explained in detail in \cite[Section~5.5]{ChasparisNatschlaeger16}, it is sufficient for the occupancy pattern signal to be persistently exciting in order for the experiments to be informative. This is indeed the case in the current experiments, due to the intermittent form of the occupancy pattern.

\subsection{Parameter Setup}	\label{sec:ParameterSetup}

For training of either the source predictor $f_{S}(\cdot)$ or the target predictor $f_{k}(\cdot)$ we employ a linear transfer model with an output-error model structure (cf.,~\cite[Section~3]{Ljung99}). In particular, the function $f$ of the prediction model (\ref{eq:GenericPredictor}) is defined as a third-order linear transfer model of the output and input/disturbance variables (i.e., $\ell=3$). The output variable is the temperature of the thermal zone under investigation, while the input/disturbance variables include the water flow, the inlet water temperature, the neighboring zone temperatures (including the outdoor temperature) and the exogenous heat disturbances. Since we want to evaluate the performance of the combined predictor over the temperature of a thermal zone, \textit{perfect estimates are assumed for all inputs and exogenous disturbances}.

The evaluation interval $T_{\rm hor}$ corresponds to a period of $6h$, which is relevant for predictions requested within an MPC implementation. As already mentioned, the sampling period was set to $T_s=1/2h$. In the computation of the performance function $R_{k}(\alpha)$ defined in (\ref{eq:MSE}), we employ a forgetting factor of $\delta=0.995$. The better reply function accepts increments of $\Delta=0.025$ in the updates of the weight $\alpha_k$. For the source predictor, we utilize a linear regression model, while for the online training of the target predictor, we employ a recursive least squares algorithm (cf.,~\cite{Sayed03}) with a forgetting factor of $0.999$. Finally, for the case of multiple source houses, (Section~\ref{sec:gotl_multiple_source}), TCA is applied with a linear kernel. The best number of principal components is selected from $n=\{5,10,15,20,25,30\}$, using 6-fold cross validation (in every fold 1/3 of one of the training scenarios is used for testing).

\section{Experiments} \label{sec:Results}

\subsection{Knowledge transfer with a single source house}	\label{sec:gotl_single_source}

We demonstrate the performance of GOTL by setting up three experiments. With progressing experiment number, the target house is chosen increasingly different from the source house. Table~\ref{tab:HouseDescription} describes the similarities between the source house and the different target houses. Weather data was collected from Washington, DC, (\emph{source house}) and Linz, Austria, (\emph{target house}) from November to March 2009. In experiment 3 the target house has also different \textit{presence patterns}. Presence patterns differ in the definition of the number of people of certain age and gender, as well as their activities (for more details, see \cite{Martinez14}).
\begin{table}[tbh]
	\centering
		\begin{tabular}{|l|l|l|l|}
		\hline
			Exp & Weather & Size & Presence \\
		\hline
		1	 &  \multirow{2}{140pt}{source: Washington, DC target: Linz, Austria}   & same                                                           &  \multirow{2}{33pt}{same}  \\ \cline{1-1} \cline{3-3} 
	  2	 &                                                                     & \multirow{2}{46pt}{target is 3x larger}                        &  \\ \cline{1-1} \cline{4-4}
 		3	 &                                                                     &                                                                &  different \\ \cline{1-1} 
		\hline
		\end{tabular}
	\caption{Differences between the Target and Source Houses.}
	\label{tab:HouseDescription}
\end{table}

The results of the experiments of Table~\ref{tab:HouseDescription} are depicted in Figure \ref{fig:results1} (Experiment 1), Figure \ref{fig:results2} (Experiment 2) and Figure \ref{fig:results3} (Experiment 3). Four predictors are evaluated: (i) \textit{Source house regressor} ($f_{S}(\cdot)$), trained on $\mathbf{X}_{S}$; (ii) \textit{Target house regressor} (${f}_{k}(\cdot)$), incrementally trained on $\mathbf{X}_{k}$; (iii) GOTL; and (iv) \textit{Ensemble Predictor}: a weighted predictor (\ref{eq:WeightedPredictor}) with a fixed $\alpha=0.5$, similarly to most ensemble methods. In all figures, we demonstrate an exponentially weighted moving average of the \emph{Root Mean Squared Error} (RMSE). 
The lower plot in Figures \ref{fig:results1}-\ref{fig:results3} always depicts the weights for the source and target house regressor at each time step $t_k$. 

\begin{figure}[h!]  
		 \centering
		  \includegraphics[width=0.75\textwidth]{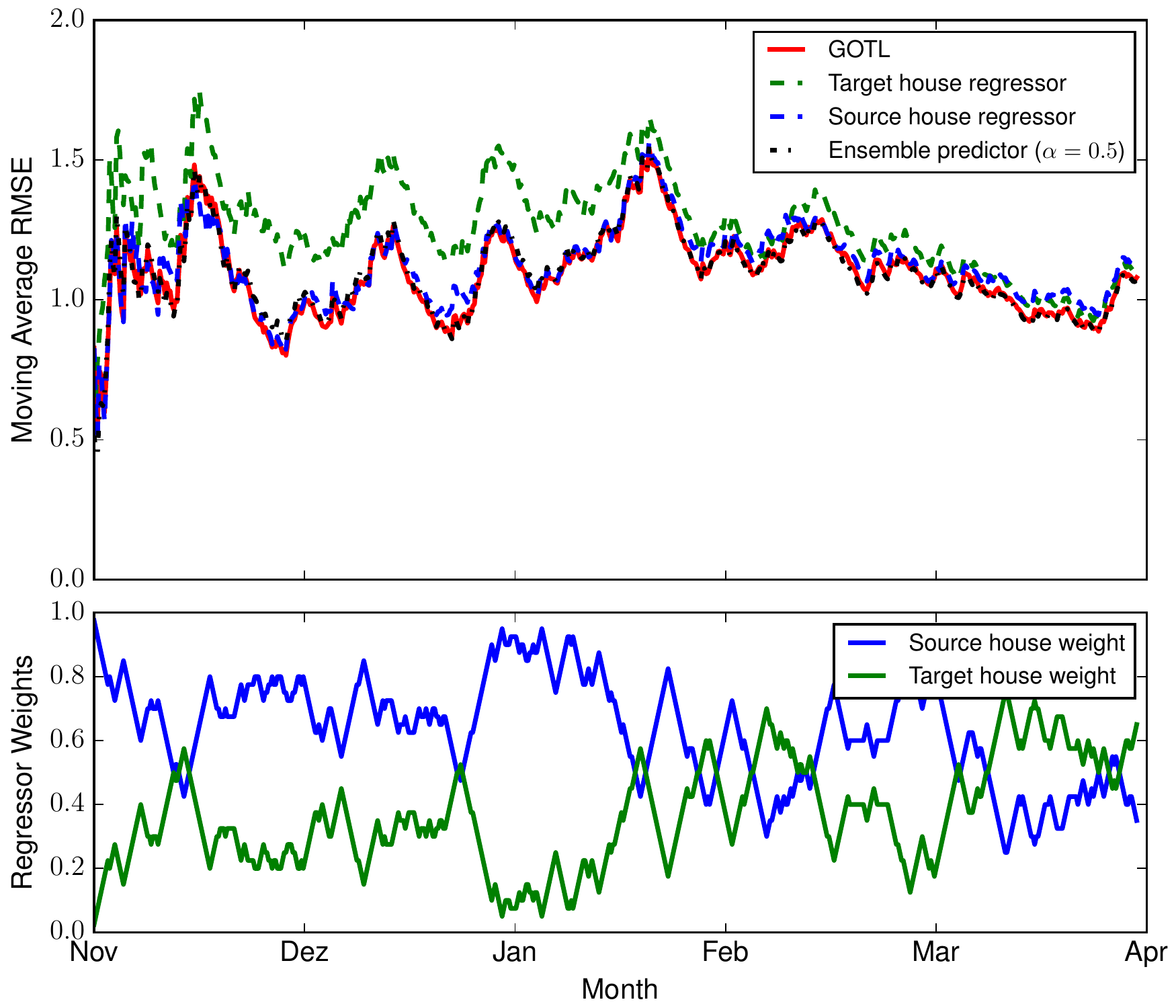}   
   \caption{Results of Experiment 1 as described in Table \ref{tab:HouseDescription}. 
   }    
		\label{fig:results1}
\end{figure}
\begin{figure}[h!]   
		\centering 
		\includegraphics[width=0.75\textwidth]{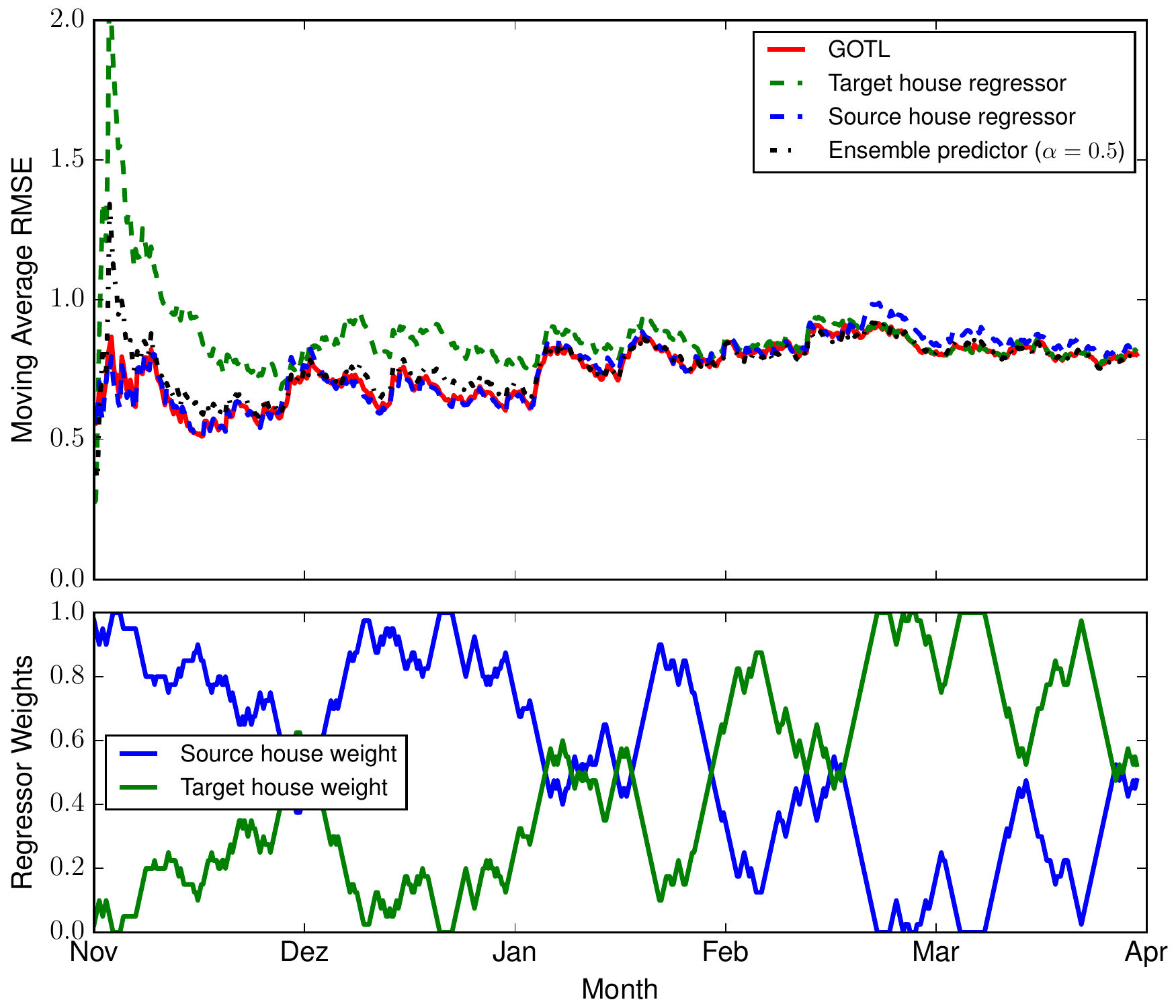}
   \caption{Results of Experiment 2 as described in Table \ref{tab:HouseDescription}. 
   }    		\label{fig:results2}	
\end{figure}
\begin{figure}[h!]   
		\centering 
		\includegraphics[width=0.75\textwidth]{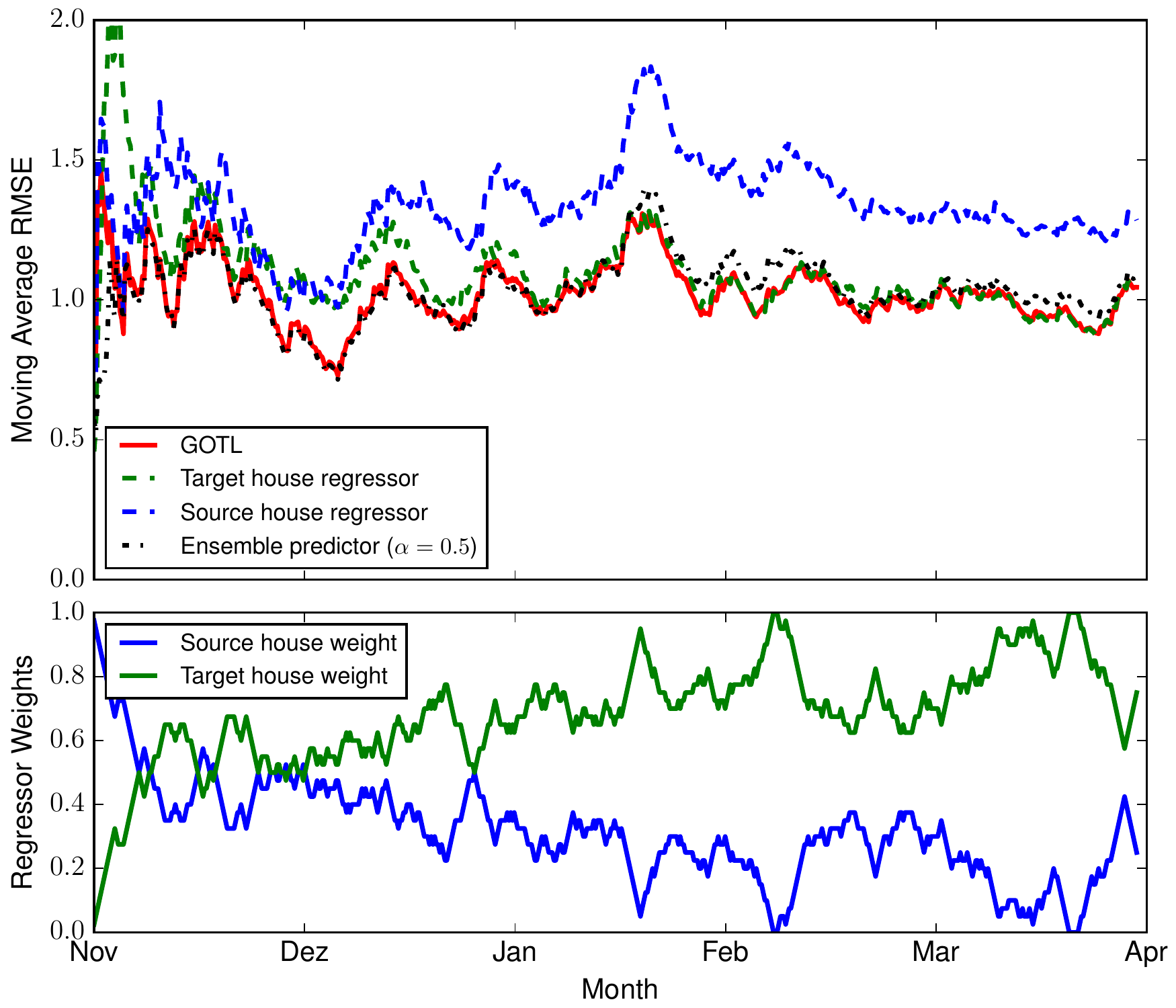}    
   \caption{Results of Experiment 3 as described in Table \ref{tab:HouseDescription}. 
   }   		\label{fig:results3}	
\end{figure}

At the beginning of the evaluation there is not enough data for the target house regressor to make good predictions. Thus, we set the initial evaluation weight to 1 for the source house regressor and 0 for the target house regressor.\footnote{Note that the coefficients for the source house regressor and the target house regressor correspond to $\alpha$ and $(1 - \alpha)$ in Section \ref{sec:OnlineTransferLearning}, respectively.} The data for the target house regressor get richer over time. Consequently, the error of the target house regressor drops and the weight of the target house regressor increases.

The observations from all three experiments can be summarized as follows: (i) In the first period -- in which the train house regressor is better than the target house regressor -- GOTL has a similar performance as the train house regressor (approx. 16 weeks in Figure \ref{fig:results1}, approx. 12 weeks in Figure \ref{fig:results2} and only 1-2 weeks in Figure \ref{fig:results3}). The much shorter time horizon in Figure \ref{fig:results1} is not unexpected as the target house is more different (the other houses use the same people presence) than the source house -- compared to the other experiments. The only difference between the Experiment 1 and 2 is the size of the house, which explains the slightly longer time horizon where the train house is better than the source house. (ii) In the second period, the weighted combination is at least as good as the target house regressor alone, or better. Particularly in Experiment 2, GOTL is much better than either the train house regressor or the test house regressor alone. In all experiments, GOTL is also at least as accurate as the simple ensemble predictor.

\subsection{Knowledge transfer with multiple source houses using TCA}	\label{sec:gotl_multiple_source}

We demonstrate that data from multiple source houses can be combined by the domain generalization method TCA and predictive accuracy can be improved -- compared to the use of only one source house. Moreover, the combined source house predictor can be combined with the target house regressor using GOTL. 

For our evaluation, we use 2 source houses. Weather data from both source houses originate from Washington, DC. \emph{Source house regressor 2} has three times the size of \textit{source house regressor 1}. The used target house has two times the size of \textit{source house regressor 1} and uses weather data from Linz, Austria. All three houses facilitate a different presence pattern.

The results are depicted in Figure~\ref{fig:results4}. It can be observed that the combined source house regressor is much better than using either one of the source houses alone. One reason for the good source house performance is that the source houses are selected, such that the target house size is just in between the sizes of two source houses. GOTL performance is comparable to the combined source house regressor in the first 5 weeks and better than all other regressors afterwards.

\begin{figure}[t!]   
		\centering 
		\includegraphics[width=0.75\textwidth]{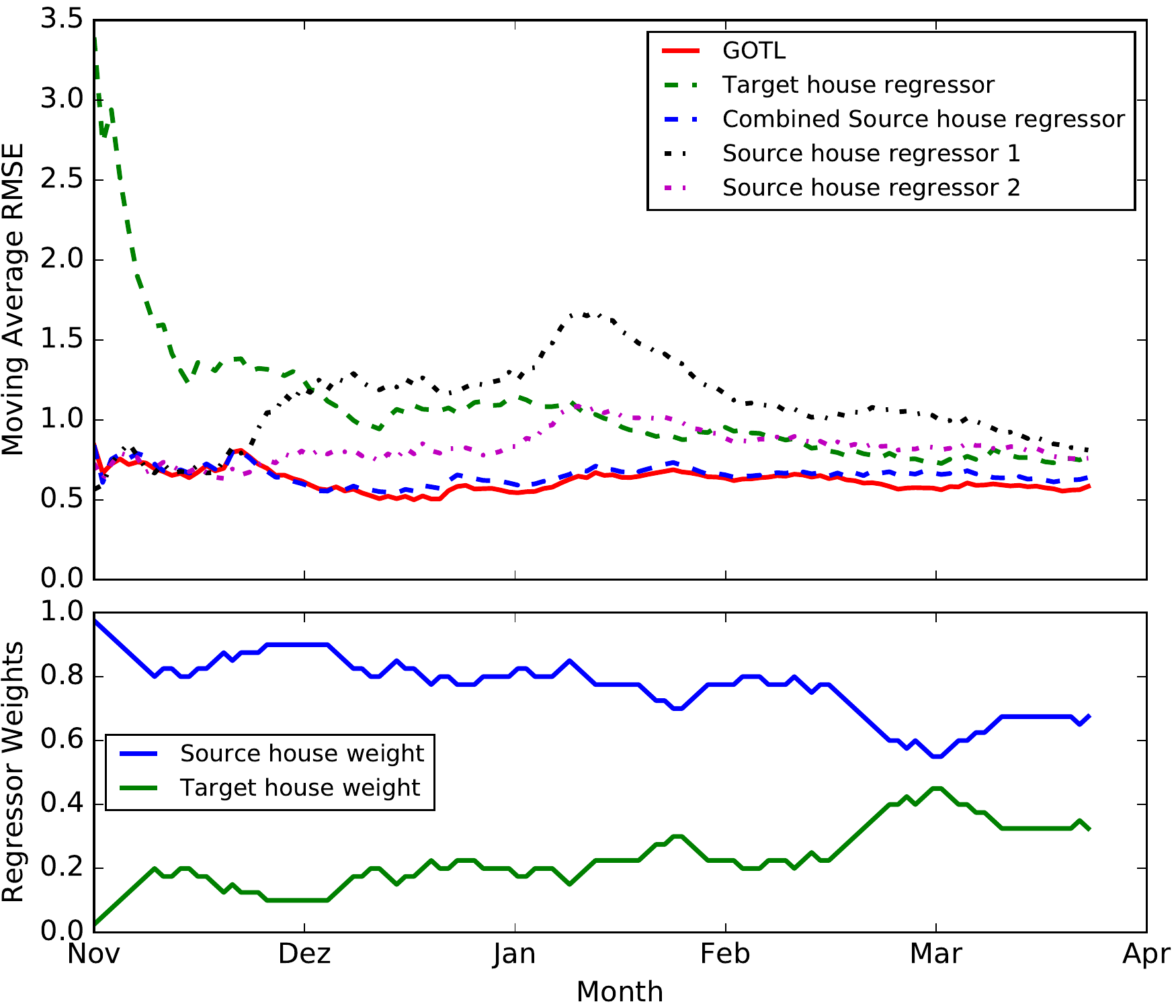}    
   \caption{Results of Experiment 4 as described in Section~\label{sec:Experiments2}}. 
   \label{fig:results4}	
\end{figure}

\subsection{MPC experiments}	\label{sec:PredictionPerformanceComparison}

In the previous experiments, we evaluated the prediction performance of the GOTL algorithm when data were generated under a standard hysteresis controller. In those experiments, the prediction models derived were not used in the controller design process, since the objective was primarily the evaluation of the prediction accuracy. In this section, instead, we wish to evaluate the performance of the introduced online transfer learning methodology when employed within an MPC formulation for climate control, i.e., when the prediction models are directly used in the control design. Questions related to whether the derived prediction models may reduce the energy consumption naturally emerge.

To this end, we designed a standard MPC for the radiant-heating system of the main living area of the residential building. The goal is to evaluate the performance of the GOTL algorithm with respect to the energy consumption and compare it with the case that such online transfer models are not available. The structure of the MPC employed is rather simple and addresses the following optimization problem
\begin{subequations}	\label{eq:MPCformulation}
\begin{align}
\min &&& \kappa\sum_{t=0}^{N_{\rm hor}}\Big\{\hat{p}(t) \left(\hat{T}_{r}(t)-T_{\rm set}(t)\right)^2/N_{\rm hor}\Big\} + \nonumber \\ 
&&& \sum_{t=0}^{N_{\rm hor}-1}\Big\{ \beta T_{s} \left(T_{w}^{+}(t)-\hat{T}_{w}^{-}(t)\right)\Big) + \gamma T_{s} \dot{V}_{w}(t)  \Big\} \label{eq:MPCObjectiveFunction} \\ 
\mbox{var.} &&& T_{w}^{+}(t)\in\{45^{o}C\}, \\
&&& \dot{V}_{w}(t)\in\{0,\dot{V}_{w,{\rm max}}\}, \\
&&& t=0,1,2,...,N_{\rm hor}-1, \nonumber
\end{align}
\end{subequations}
where $\hat{T}_{r}$ is the temperature prediction of the room provided by the available prediction model (e.g., the GOTL model), $T_{\rm set}$ is the desired/set temperature, $T_{w}^{+}$ is the inlet water temperature of the radiant heating system, $\hat{T}_{w}^{-}$ is a prediction of the outlet water temperature, and $\dot{V}_{w}$ is the water flow.

Note that the first part of the objective function (\ref{eq:MPCObjectiveFunction}) corresponds to a \emph{comfort} measure scaled with a positive constant $\kappa$. It measures the average squared difference of the room temperature from the desired (or set) temperature entered by the user at time $k$. The set temperature was set equal to $21^oC$ throughout the optimization horizon. The variable $\hat{p}\in\{0,1\}$ is a boolean variable that corresponds to our estimate on whether people are present in the room at time instance $t$. 

The second part of the objective function (\ref{eq:MPCObjectiveFunction}) corresponds to the \emph{heating cost}, while the third part corresponds to the \emph{pump-electricity cost}. The nonnegative parameters $\beta$, $\gamma$ were previously identified for the heating system of the simulated house and take values: $\beta= 0.3333 kW/^oCh$, $\gamma=0.5278\cdot{10}^{3} kWsec/h m^3 $. The non-negative constant $\kappa$ is introduced to allow for adjusting the importance of the comfort cost compared to the energy cost. 

The $\hat{p}(t)$, $t=1,2,...,N_{\rm hor}$, as well as the rest of the disturbances (such as the outdoor temperature and the solar gain) are assumed given (i.e., predicted with perfect accuracy). This assumption is essential in order to evaluate precisely the impact of our prediction model $\hat{T}_{r}$ in the performance of the derived optimal controller. In fact, we should expect that an accurate prediction of the room temperature will also lead to an efficient controller with respect to the energy consumption.

The sampling period was set to $T_{\rm s} = \nicefrac{1}{2}h$, the optimization period was set to $T_{\rm opt}=1h$, and the evaluation/optimization horizon was set to $T_{\rm hor}=6h$. This implies that $N_{\rm hor}=6 \cdot 2 = 12$. Furthermore, the control variables are the inlet water temperature which assumes a single value ($45^oC$) and the water flow which assumes only two values, the zero flow and the maximum one, $\dot{V}_{w,\max}=0.0787kg/sec$. 

In Figure~\ref{fig:consumption_comfort}, we have generated the comfort-heating cost curve for the available prediction models under varying $\kappa$ (i.e., under different weights on the comfort cost of the objective function (\ref{eq:MPCObjectiveFunction})). The reason for generating such performance curves is the fact that the performance of any two controllers with respect to energy efficiency may \emph{only} be compared under the same comfort-cost level. Of course, a controller that achieves lower heating costs under the same comfort cost will be more efficient. Experiments were conducted over a five-month period under Experiment~3 of Table~\ref{tab:HouseDescription}. The prediction model of the target house was trained online with the data generated under the MPC implementation, using the recursive least squares of Section~\ref{sec:ParameterSetup}. The weights $\alpha$ of the GOTL predictor were also updated online according to the online algorithm of Table~\ref{Tb:GOTL}. 

As we observe in Figure~\ref{fig:consumption_comfort}, the heating costs under the GOTL predictor are lower than any other prediction model. Under high comfort costs, i.e., when comfort is not of high priority, the energy consumption is very close to the energy consumption of the train model. On the other hand, under low comfort costs, i.e., when comfort is a priority, the energy consumption of the GOTL predictor is about 100 kWh lower compared to the trained predictor (over the five-month period), which is a non-negligible amount of energy. It is also important to note that the performance of the test predictor with respect to the energy consumption is rather poor for these first five months. This verifies our claim over the utility of online transfer learning methodologies for climate control in residential buildings.

\begin{figure}[t!]   
		\centering 
		\includegraphics[width=0.8\textwidth]{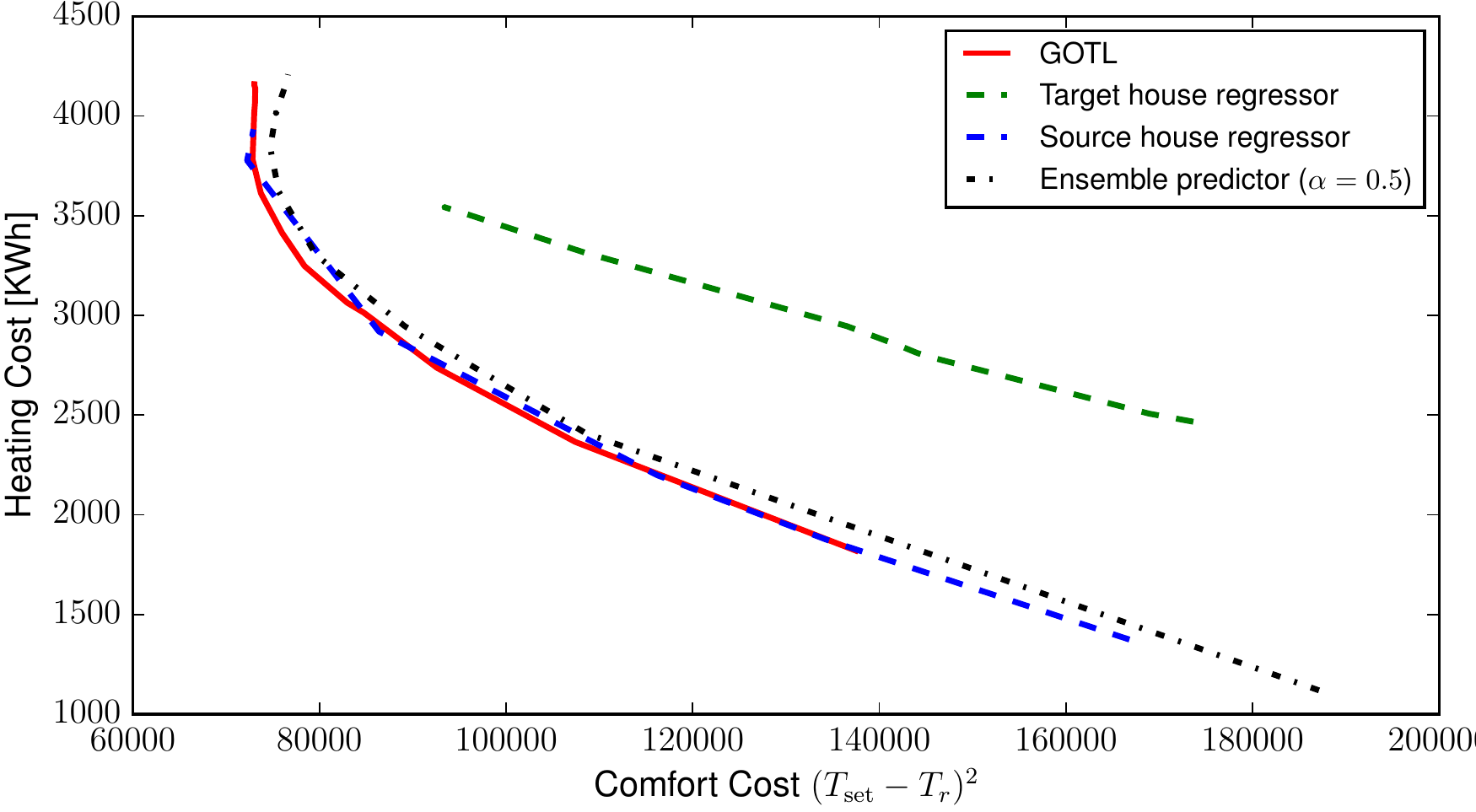}    
   \caption{MPC experiments}. 
   \label{fig:consumption_comfort}	
\end{figure}


\begin{figure}[th!]
\centering
\end{figure}

\section{Conclusions \& Future Work} \label{sec:Conclusions}


We presented the online transfer learning framework GOTL, which is applicable for classification and regression tasks and allows to optimally combine an (offline) source domain predictor with an online target domain predictor. The results demonstrated the utility of the combined predictor to significantly improve prediction accuracy in the first weeks and months of a new building -- compared to either using the source house predictor or target house predictor alone. Further improvements in predictive accuracy can be achieved by facilitating multiple source houses and TCA. Improvements in predictive accuracy also translate into non-negligible energy savings for given comfort levels. 

It is also important to note the adaptive response of the proposed online transfer learning algorithm. Although the above experiments were evaluated over a period of five winter months, training of an online regressor over longer periods of time may be challenging. This is primarily due to changes either in the weather conditions or in the heating patterns. In such variations, degradations in the prediction error of a target predictor will be significantly reduced by considering the GOTL algorithm, as can easily be seen from the behavior of GOTL during the first few weeks in all the considered experiments.



\section{Bibliography}  \label{sec:Bibliography}
\bibliographystyle{elsarticle-num} 
\bibliography{2016_EB_GOTL_Bibliography}


\end{document}